\documentclass[letterpaper, 10 pt, conference]{ieeeconf}

\IEEEoverridecommandlockouts

\usepackage{graphics} 
\usepackage{epsfig} 
\usepackage{mathptmx} 
\usepackage{times} 
\usepackage{amsmath} 
\usepackage{amssymb}  
\usepackage{txfonts}
\usepackage{bbm}
\usepackage[singlelinecheck=false,font=small]{caption}
\usepackage{subcaption}
\usepackage{color}  
\usepackage[section]{placeins}

\newtheorem{defn}{Definition}
\newtheorem{rem}[defn]{Remarks}
\newtheorem{lem}[defn]{Lemma}

\newtheorem{assum}[defn]{Assumption}
\newtheorem{ex}[defn]{Example}
\newtheorem{thm}[defn]{Theorem}

\providecommand{\R}{\ensuremath \mathbb{R}}
\providecommand{\N}{\ensuremath \mathbb{N}}

\providecommand{\X}{\ensuremath \mathcal{X}}

\providecommand{\frs}{_\text{FRS}}

\newcommand{\Lf}{\mathcal{L}_{f_s}}
\newcommand{\Lg}{\mathcal{L}_g}

\title{Safe Trajectory Synthesis for Autonomous Driving in Unforeseen Environments}

\author{Shreyas Kousik$^1$, Sean Vaskov$^1$, Matthew Johnson-Roberson$^2$, Ram Vasudevan$^1$%
\thanks{$1$. University of Michigan - Ann Arbor, Department of Mechanical Engineering. Contact \texttt{$\{$skousik, skvaskov, ramv$\}$@umich.edu}.\newline
\indent$2$. University of Michigan - Ann Arbor, Department of Naval Architecture and Marine Engineering. Contact \texttt{mattjr@umich.edu}.}}

\begin{document}
\maketitle
\thispagestyle{empty}
\pagestyle{empty}

\begin{abstract}
Path planning for autonomous vehicles in arbitrary environments requires a guarantee of safety, but this can be impractical to ensure in real-time when the vehicle is described with a high-fidelity model.
To address this problem, this paper develops a method to perform trajectory design by considering a low-fidelity model that accounts for model mismatch.
The presented method begins by computing a conservative Forward Reachable Set (FRS) of a high-fidelity model's trajectories produced when tracking trajectories of a low-fidelity model over a finite time horizon.
At runtime, the vehicle intersects this FRS with obstacles in the environment to eliminate trajectories that can lead to a collision, then selects an optimal plan from the remaining safe set.
By bounding the time for this set intersection and subsequent path selection, this paper proves a lower bound for the FRS time horizon and sensing horizon to guarantee safety.
This method is demonstrated in simulation using a kinematic Dubin's car as the low-fidelity model and a dynamic unicycle as the high-fidelity model.
\end{abstract}

\section{INTRODUCTION}
\label{sec:introduction}

Safe control of autonomous vehicles for use on public roads is a challenging problem that directly impacts both public safety and development costs for auto manufacturers. 
To be safe, autonomous vehicles must be able to account for errors in their sensors and dynamic models.
Moreover, since it is impractical to preplan an obstacle avoidance trajectory for every plausible environment and configuration of obstacles, controllers need to be constructed, evaluated for correctness, and selected in real-time. 

Autonomous vehicles typically apply a hierarchical control design structure with three levels \cite{falcone2007predictive,buehler2009darpa,gray2012predictive}. 
The highest-layer in this architecture is responsible for route planning using Dijkstra's algorithm or its variants. 
The mid-level controller generates trajectories to perform actions such as lane-keeping or obstacle avoidance. 
The low-level controller follows these generated trajectories by applying throttle, brake, and steering inputs. 
For computational efficiency, the mid-level controller often uses a lower degree of freedom model than the low-level controller. 
As a result, there is typically a non-trivial gap between the trajectories generated by the mid-level controller and what a vehicle is capable of executing.
This paper develops a real-time optimization based scheme for the mid-level controller that can provably design trajectories which can be safely followed by the low-level controller in arbitrary environments with static obstacles.

Various methods exist in the literature for mid-level controller design, which rely upon either temporal or spatial discretization, or both.
Rapidly Random Exploring Trees, for example, compute safe trajectories by sampling from a vehicle's input space and numerically integrating trajectories forward using a dynamic model \cite{lavalle2001randomized}. 
These methods can handle nonlinear dynamics and non-convex constraints and can even be made asymptotically optimal \cite{karaman2011sampling}.
Partially Observable Markov Decision Processes are also able to handle nonlinear dynamics, non-convex constraints, and bounded uncertainty, while planning trajectories in a spatially and temporally discretized space\cite{brechtel2014probabilistic}. 
Though each of these mid-level control design methods is powerful, they struggle to account for the uncertainty introduced by discretization, which can result in trajectories that are impossible to execute.
For this reason, others have focused solely on the real-time verification of the safety of a generated mid-level controller using approaches grounded in zonotopes \cite{althoff2014online}.

\begin{figure}
\includegraphics[width=1\columnwidth]{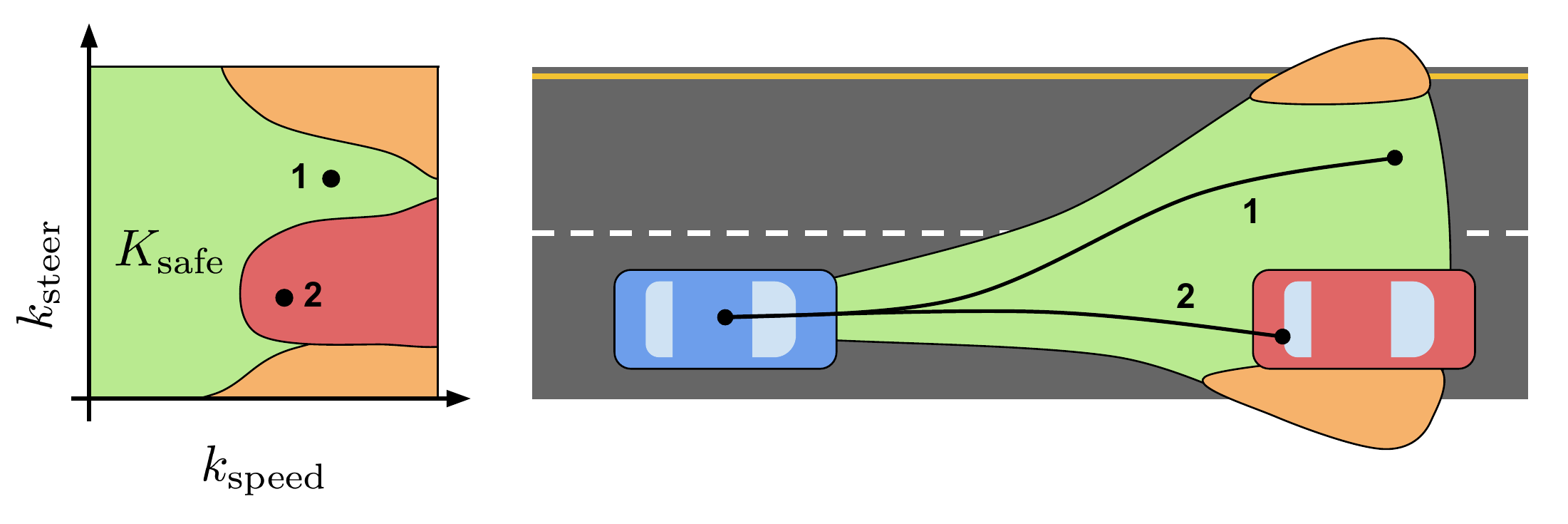}
\caption{An illustration of the Forward Reachable Set in the physical world (right) and the trajectory parameter space (left). 
Each trajectory parameter on the left corresponds to a trajectory in the physical world. 
The FRS (the funnel beginning from the blue car) is intersected with obstacles (the road boundaries and the red car) at run-time to identify trajectory parameters that lead to collisions (orange and red sets in left picture). 
In this sketch, the parameters labeled 1 produce a guaranteed-safe trajectory, whereas the parameters labeled 2 result in a collision.
Optimal parameters can be selected from the safe set, $K_\text{safe}$.}
\label{fig:overallsketch}
\end{figure}

Several approaches have recently been pursued to perform control design with safety guarantees.
Robust Model Predictive Control (MPC), for example, has been demonstrated to successfully account for uncertainty in the dynamics and can be used for safe control design for nonlinear vehicle models in the presence of non-convex constraints \cite{gao2014tube,shia2014semiautonomous}.
This requires linearizing the dynamics about a pre-specified, spatially discretized trajectory to synthesize a controller that is able to safely follow this trajectory while avoiding obstacles.
Thus, this approach requires solving a scenario-specific nonlinear program, which can be challenging to solve at run-time and, as a result, has few associated safety guarantees.
To try to address the more general safe control design problem, others have employed the Hamilton-Jacobi-Bellman Equation, which is typically solved using state-space discretization, to compute the set of safe feedback controllers in the presence of bounded disturbance \cite{ding2011reachability}.
These approaches rely upon a complete characterization of obstacle locations, which can make real-time applications challenging due to the computational expense associated with state-space discretization.
Despite this limitation, extensions of this approach have been successfully employed to perform cooperative control for connected vehicles \cite{dhinakaran2017hybrid}.
Similarly, the viability kernel, which is also computed using state-space discretization, has been employed to evaluate the set of possible trajectories for autonomous vehicles on pre-defined scenarios \cite{liniger2017real}. 

To avoid this curse of dimensionality, others have employed barrier functions using Lyapunov theory to devise safety-preserving controllers for adaptive cruise control and lane-keeping \cite{nilsson2016correct,xu2016correctness,wang2017safety}.
These approaches can quickly solve a quadratic program at run-time to guarantee safety, but the barrier function, which is computed off-line, is scenario specific. 
To overcome this situation-specific limitation, an approach relying upon funnel libraries was recently proposed to perform real-time control design in the presence of uncertainty \cite{majumdar2016funnel}.
This approach relies upon pre-computing a rich-enough finite family of trajectories, which is then searched at run-time to ensure safety. 


This paper presents an approach for mid-level control design that simultaneously accounts for nonlinear dynamics, non-convex constraints, and uncertainty without pre-specifying the scenario, pre-computing a finite family of trajectories, or relying on spatial or temporal discretization. 
The proposed method, which is depicted in Figure \ref{fig:overallsketch}, begins by computing the set of trajectories that can be reached by a high-fidelity vehicle model under a parameterized controller.
This Forward Reachable Set (FRS), which is computed off-line, is then intersected with obstacles at run-time using convex optimization to exclude controllers that could cause collisions. 
Finally, the remaining safe set of controllers is optimized over to minimize an arbitrary cost function, such as following a nominal trajectory, minimizing total acceleration, etc.

This paper makes the following pair of contributions:
first, a convex optimization based method to synthesize trajectories which can be safely tracked by a high fidelity vehicle model;
and second, a method to select braking, sensing, and planning time horizons to guarantee that a safety preserving trajectory always exists in an environment with static obstacles.
The rest of the paper is structured as follows: 
Section \ref{sec:notation} introduces the notation and assumptions utilized throughout the paper.
Section \ref{sec:problem_formulation} outlines methods to determine the FRS in the presence of bounded uncertainty, to compute the set of safe trajectory parameters in the presence of obstacles, and to compute an optimal trajectory from this set.
Section \ref{sec:implementation} describes the implementation of the methods formulated in Section \ref{sec:problem_formulation}. 
Section \ref{sec:examples} illustrates the performance of our method on an example.

\section{PRELIMINARIES}
\label{sec:notation}

This section introduces necessary notation for this paper, and formalizes the dynamics and assumptions using to describe the problem of interest.
This paper makes substantial use of semidefinite programming theory and sums of squares programming \cite{lasserre2009moments}.

\subsection{Notation}

The following notation is adopted throughout the remainder of the paper: 
Sets are italicized and capitalized.
The boundary of a set $K$ is $\partial K$.
Finite truncations of the set of natural numbers are \mbox{$\N_n:=\{1,\ldots,n\}$}.
The set of continuous functions on a compact set $K$ are $\mathcal C(K)$.
The Lebesgue measure on $K$ is denoted by $\lambda_K$.
We let $y_{K}$ denote the Lebesgue moments of $\lambda_K$.
The $i$-th component of a vector $v \in \R^n$ is denoted by $v_i$.
The ring of polynomials in $x$ is $\R[x]$, and the degree of a polynomial is the degree of its largest multinomial; 
the degree of the multinomial $x^\alpha,\,\alpha\in \N^n$ is $|\alpha|=\|\alpha\|_1$; and $\R_d[x]$ is the set of polynomials in $x$ with degree $d$.
Let $\text{vec}(p)$ denote the vector of coefficients of $p \in \R[x]$.

\subsection{Low and High-Fidelity Vehicle Models}

Let an autonomous vehicle by described by:
\begin{equation}
\label{eq:big_dyn}
\dot{x}(t) = f(t,x(t),u(t))
\end{equation}
where $f:[0,T] \times X \times U \to \R^n$, $T > 0$, $ X \subset \R^n$ and $U \subset \R^m$. 
The objective of this paper is to design safe trajectories for this autonomous vehicle which may be high dimensional and challenging to optimize over in real-time.
Our approach relies upon generating parameterized trajectories using a second, lower-dimensional model for this autonomous vehicle that shares a common set of states, $X_s \subset \R^{n_s}$ with $X_s \subset X$.
This model is described as:
\begin{equation}
\label{eq:small_dyn}
\begin{bmatrix}\dot{x}_s(t) \\ \dot{k}(t) \end{bmatrix} = \begin{bmatrix}  f_s(t,x_s(t),k(t)) \\ 0 \end{bmatrix}
\end{equation}
where the parameters, $k(t)$, do not evolve after initialization and are drawn from a set $K \subset \R^p$. 

\begin{ex}
\label{ex:model}
Consider a dynamic unicycle model:
\begin{equation}
\begin{bmatrix} \dot{x}(t) \\ \dot{y}(t) \\ \dot{\theta}(t) \\ \ddot{\theta}(t) \\ \dot{v}(t) \end{bmatrix} = \begin{bmatrix} v(t)\cos\left(\theta(t)\right) \\ v(t)\sin\left(\theta(t)\right) \\ \dot{\theta}(t) \\ u_1(t) \\ u_2(t) \end{bmatrix}
\end{equation}	
where $x$ and $y$ represent position; $\theta$ and $\dot{\theta}$ represent heading and yaw rate, both relative to the initial condition of the vehicle; and $v$ represents the longitudinal speed. 
The inputs are $u_1$ for steering and $u_2$ for throttle or braking.
Consider a low-fidelity, trajectory-producing model with shared states $\begin{bmatrix} x & y & \theta \end{bmatrix}^T$ such that:
\begin{equation}\label{eq:dubins_dyn}
\begin{bmatrix} \dot{x}(t) \\ \dot{y}(t) \\ \dot{\theta}(t) \end{bmatrix} = \begin{bmatrix} k_2\cos\left(\theta(t)\right) \\ k_2\sin\left(\theta(t)\right) \\ k_1 \end{bmatrix}
\end{equation}
with the trajectory parameters that are steering rate $k_1$ and speed $k_2$.
The unicycle model ``adds'' mass to the Dubins car.
\end{ex}

For the sake of convenience, the dynamics of the shared states, $X_S$, in the high dimensional model are assumed to be described by the first $n_s$ coordinates of $f$ throughout this paper. 
Let the controlled vehicle be represented as a compact set $X_0 \subset X_s$, typically a bounding box around the vehicle.
We make the following assumptions to ensure the well-posedness of the subsequent problem formulation. 
\begin{assum} \label{ass:set}
The sets $U$ and $X$ are compact and suppose $K,X_0,$ and $X_s$ have the following representation:
\begin{align}
K &= \left\{ k \in \R^p \mid h_{K_i}(k) \geq 0, ~ \forall i \in \{ 1,\ldots, n_K\} \right\}, \label{eq:ass_2_1}\\
X_0 &= \left\{ x_s \in \R^{n_s} \mid h_{0_i}(x_s) \geq 0, ~ \forall i \in \{ 1,\ldots, n_0\} \right\},  \label{eq:ass_2_2} \\
X_s &= \left\{ x_s \in \R^{n_s} \mid h_{X_i}(x_s) \geq 0, ~ \forall i \in \{ 1,\ldots, n_{X_s}\} \right\}. \label{eq:ass_2_3}
\end{align}
\end{assum}

The objective of this paper is to optimize over trajectories using the lower complexity dynamical model which are then followed by the high dimensional vehicle model. 
Unfortunately this tracking cannot be done perfectly. 
The higher fidelity model can select a controller, $u_k: [0,T] \times X \to \R^m$, to follow the generated trajectory which is parameterized by $k \in K$ (e.g. using a PID control loop); however there may still be a gap between the generated pair of trajectories.
To describe this error as a function of time, we make the following assumption:
\begin{assum}\label{ass:error_func}
There exists a bounded function $g: [0,T] \times X_s \to X_s $ such that:
\begin{align}\label{eq:error_func}
\max_{x \in A_{x_s}} \, \left|f_i(t,x,u_k(t,x) - f_{s,i}(t,x_s,k)\right| \leq g_i(t,x_s),
\end{align}
for all $x_s \in X_s$, $t \in [0,T]$, $k \in K$, and $i \in \{1,\ldots,n_s \}$, where $A_{x_s} := \{ x \in X \mid x_i = x_{s,i} \forall i \in \{1,\ldots,n_s \} \}$.
\end{assum}
In other words, the error in the shared states is bounded by $g$.
Though this paper does not describe a formal method to construct such a $g$, in the instance of polynomial, rational, or trigonometric dynamics, such a $g$ can be found numerically by applying Sums-of-Squares programming as we describe in Section \ref{sec:implementation}. 

\begin{ex}
For the unicycle model in Example \ref{ex:model}, $u_k$ can be a proportional velocity controller for the yaw rate and speed:
\begin{align}
\ddot{\theta}(t) = u_1(t) = 20\cdot(\dot{\theta}(t) - \dot{\theta}_\text{des}) \\
\dot{v}(t) = u_2(t) = 10\cdot(v(t) - v_\text{des})
\end{align}
where $(\dot{\theta}_\text{des}, v_\text{des})$ map to a $k_1$ and $k_2$ in the parameter space for the Dubins car model described in Equation \eqref{eq:dubins_dyn}.
Given this $u_k$ one can generate a $g$ that overapproximates the error when transitioning from top speed to a complete stop in the $v$ dimension; and when switching from turning full-left to turning full-right (or vice-versa) in the $\dot{\theta}$ dimension. 
Such a $g$, illustrated in Figure \ref{fig:error_dynamics}, could be:
\begin{equation}
\label{eq:g_definition}
g(t,x,y,\theta) = \begin{bmatrix}
v_\text{err}(t)\cdot(1 - \frac{1}{2}\theta^2)  \\
v_\text{err}(t)\cdot(\theta - \frac{1}{6}\theta^3) \\
\dot{\theta}_\text{err}
\end{bmatrix}
\end{equation}
where $v_\text{err}(t) = (t-1)^2$ and $\dot{\theta}_\text{err}(t) = (t-1)^4$.
\end{ex}

\begin{figure}
\centering
\includegraphics[width=1\columnwidth]{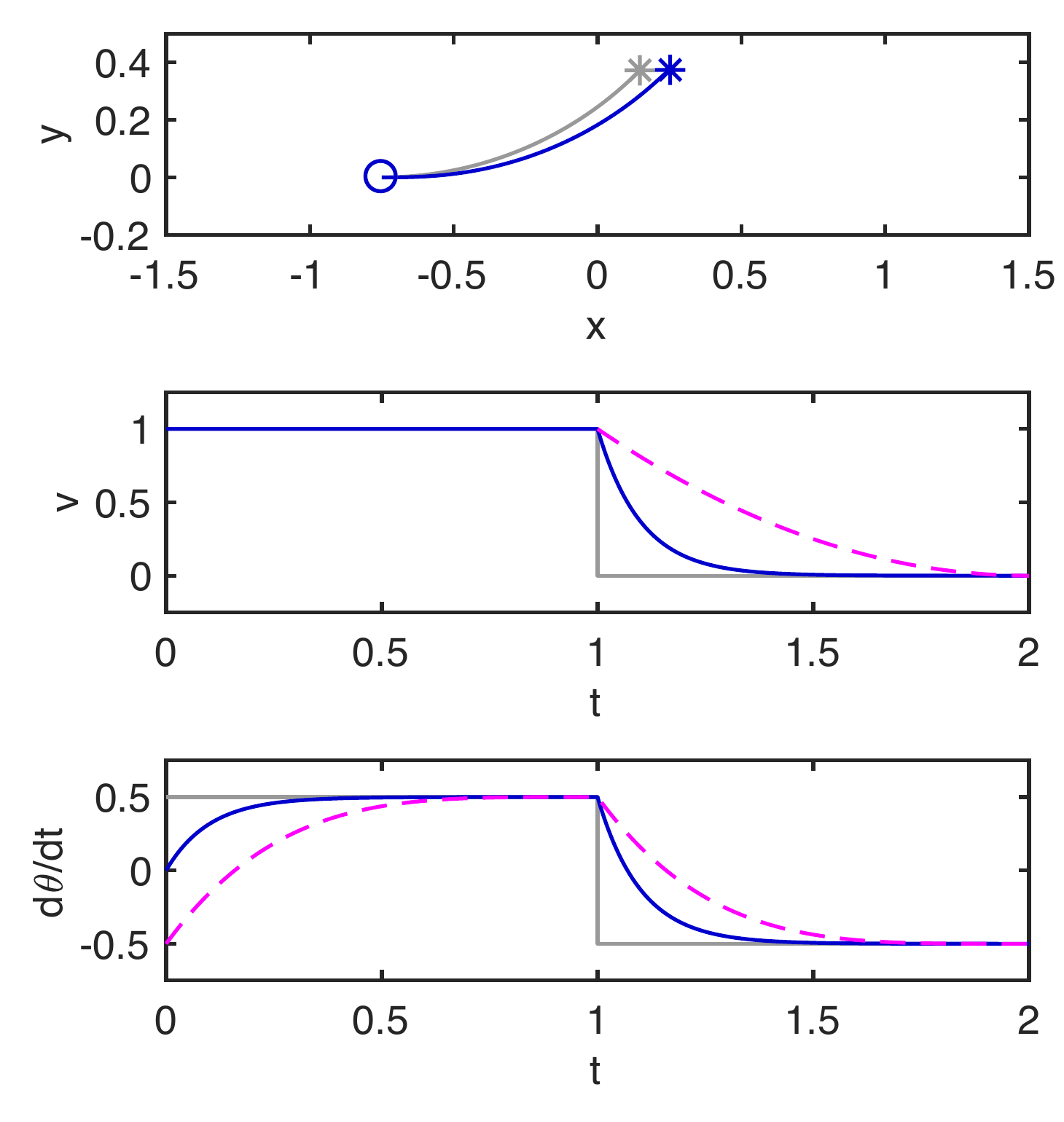}
\caption{An example of the worst-case error dynamics. 
The top subplot shows the $(x,y)$ trajectories of a Dubins car (grey) and unicycle model (blue), beginning at $(-0.75,0)$ and ending at the corresponding stars. Both trajectories begin at full speed (1 [m/s]) and then emergency brake at $t = 1$. The resulting velocity profile, with the same color scheme, is shown in the middle subplot. In addition, the Dubins car begins its trajectory with $\dot{\theta} = +0.5$ [rad/s], whereas the unicycle begins at $0$ [rad/s] and must catch up. At $t = 1$, the Dubins car switches instantaneously to turning at $-0.5$ [rad/s], as shown in the bottom subplot. The red dashed line in the middle and bottom subplots is the error function $g$ that overapproximates the worst-case possible error between the Dubins car and unicycle with a polynomial.}
\label{fig:error_dynamics}
\end{figure}

As a result of this assumption, the dynamics can be rewritten in a form amenable to computing the flow of the system $f_s$ while subject to error dynamics:
\begin{align}
\label{eq:error_low_model}
\dot{x}_s(t) = f_s(t,x_s(t),k) + g(t,x_s(t),k)~d(t)
\end{align}
where $d(t) \in [-1,1]$ for each $ t \in [0,T]$ can be chosen to describe the worst-case error behavior.
This use of $d$ is discussed further in Section \ref{subsec:FRS}.
To understand the gap between lower and higher dimensional models, we rely upon a pair of linear operators. 
First, the linear operator $\Lf:C^1\big([0,T] \times X_s \times K \big) \to C\big( [0,T] \times X_s \times K \big)$ on a test function $v$ as:
\begin{equation}
\Lf v(t,x_s,k) = \frac{\partial v}{\partial t}(t,x_s,k) + \sum_{i = 1}^{n_s} \frac{\partial v}{\partial x_{i,s}}(t,x_s,k) f_{i,s}(t,x_s,k).
\end{equation}
Next, define the linear operator $\Lg:C^1\big([0,T] \times X_s \times K \big) \to C\big( [0,T]\times X_s \times K\big)$ as:
\begin{equation}
\Lg v(t,x_s,k)  = \sum_{i = 1}^{n_s} \frac{\partial v}{\partial x_{i,s}}(t,x_s,k) g_{i}(t,x_s).
\end{equation}
Note that these linear operators are useful in summarizing the evolution of a system as we describe in Section \ref{subsubsec:FRS_computation}.



\subsection{Sensing, Planning, and Braking Time Horizons}
\label{subsec:time}

Next, we define a sensing, planning, and braking time horizon whose relationships are formalized in the next section to ensure the safety of the autonomous vehicle. 
To construct these definitions, we begin by formalizing obstacles: 
\begin{defn}
An \emph{obstacle} is any portion of $X_s$ that must be avoided by the vehicle.
In general, obstacles are represented as closed subsets of $X_s$.
Obstacles are assumed to be static in the system's local reference frame and are defined as:
\begin{equation}
X_{\text{obs}} = \left\{ x_s \in \R^{n_s} \mid h_{\text{obs}_i}(x_s) \geq 0 , \forall i \in \{1,\ldots,n_{\text{obs}} \} \right\}.
\label{eq:X_obs}
\end{equation}
\end{defn}
Obstacles can be other vehicles, the edges of the road, pedestrians, etc.
A trajectory is \emph{safe} if it does not intersect with an obstacle.
These obstacles are assumed to enter the scene as follows:
\begin{assum}\label{ass:sense}
During operation, obstacles always enter $X_s$ from outside, i.e., any obstacle's trajectory relative to the system must begin outside of, and may pass through, $\partial X_s$.
In addition, let the system be limited to a scalar top speed, denoted $\dot{x}_{s,\text{max}}$.
Then, this assumption requires that obstacles are sensed at a minimum sensor horizon $D_\text{sense} := \dot{x}_{s,\text{max}}\cdot T_\text{sense}$.
\end{assum}
$T_\text{sense}$ is a sensor time horizon that is to be determined, as a safety condition that we describe in the next section. 

Next, we assume that the time to process sensor information and plan a trajectory is bounded and known:
\begin{assum}\label{ass:plan_and_scan}
The time required to process sensor data and perform trajectory planning has a finite upper bound $\tau_\text{plan}$.
\end{assum}
Specific sensor processing to detect obstacles is outside of the scope of this paper, but in practice most modern obstacle detectors, working from camera, lidar, or radar, have a bounded processing time \cite{johnson2016driving,liu2016ssd}.
Though we do not prove that the trajectory planning time is bounded, the convex optimization based implementation that we utilize has a processing time that takes several seconds, at most, in practice as we describe in Section \ref{sec:implementation}.

\begin{rem}\label{rem:plan}
To understand how this $\tau_\text{plan}$ can be used, suppose the vehicle is planning as quickly as possible and executing a trajectory beginning at $t=0$ . 
While executing this trajectory, the vehicle can plan its next trajectory which is applied beginning at $t = \tau_{\text{plan}}$. 
The initial condition for this next trajectory is computed by the vehicle by forward integrating $f$ over the time interval $[0,\tau_{\text{plan}}]$.
This process can be repeated recursively to replan trajectories. 
\end{rem}


Finally, in certain instances it may be impossible to reach a desired position safely.
For these scenarios, we must assume that there exists a braking maneuver which can stop the car:
\begin{assum}\label{ass:brake}
There exists a family of braking trajectory parameters $K_b \subseteq K$ such that, without deviating from the spatial component of a pre-planned trajectory, $\dot{x}_s = 0$ in a finite amount of time $\tau_b(x_s) \in \R_+$ for each $x_s \in X_s$.
In addition, this stopped state can be held indefinitely.
\end{assum}
In other words, the vehicle is able to brake and stop along any trajectory it is following.
It follows from Assumption \ref{ass:brake} that there is a maximum braking time for the system, i.e. there exists $\tau_\text{stop} \in \R_+$ such that $\tau_b(x_s) \leq \tau_\text{stop}$ for all $x_s \in X_s$.
This braking time is used to construct planning and sensing time horizons to guarantee safety.


\section{PROBLEM FORMULATION}
\label{sec:problem_formulation}

This section outlines a method for trajectory design and presents a theorem relating the braking, planning, and sensing time horizons to ensure the safety of the trajectory design procedure.
The approach is broken into the following steps:
\begin{enumerate}
\item Precompute a \emph{Forward Reachable Set} (FRS) that captures all possible trajectories and associated parameters, subject to error function $g$, over a time interval $[0,T]$ (Section \ref{subsec:FRS}).
\item During operation, perform \emph{set intersection} of the FRS with obstacles in the vehicle's local environment to identify trajectory parameters, $K_{\text{safe}} \subset K$, which could cause collisions using a convex optimization technique (Section \ref{subsec:set_intersection}).
\item During operation, perform \emph{trajectory optimization} over $K_{\text{safe}}$ of a user-specified cost function to drive the system towards an objective while guaranteeing safety (Section \ref{subsec:traj_opt}).
\end{enumerate}
These steps are described next.

\subsection{Forward Reachable Set}
\label{subsec:FRS}

This subsection describes our approach to compute the FRS, which contains all points that can be reached from $X_0$ under the dynamics of $f$. 
The dynamics of $f$ may be high-dimensional in general as a result, which can make computing the FRS for $f$ impractical.
To overcome this challenge, we compute the FRS of the lower-dimensional model described in Equation \eqref{eq:error_low_model}.
\begin{equation}
\begin{aligned}
\X\frs = \Big\{(x_s,k) \in X_s \times K \ | \; &  \exists~ x_0 \in X_0, \tau \in [0,T], \text{ and } \\ 
& d: [0,T] \to [-1,1] \text{ s.t. } x(\tau) = x_s, \text{ where} \\
& \dot{x}(t) = f_s(t,x(t),k) + g(t,x(t),k)\; d(t)\\
& \text{a.e. } t \in [0,T] \text{ and } x(0) = x_0 \Big\}
\end{aligned}
\end{equation}
Notice that in this definition the error function $g$ along with $d(t) \in [-1,1]$ accounts for the difference between a trajectory, parameterized by $k \in K$, of the lower dimensional model and the higher-dimensional model which follows this trajectory using a control $u_k$.


\subsubsection{Selecting $T$ to Ensure Safety}

Identifying the time horizon $T$ for the FRS is the critical first step.
If $T$ is too short, then safe behavior cannot be guaranteed; if it is too long, then computing the FRS may be impractical.
Before describing how to compute $\X\frs$, we first describe how to select the time horizon $T$ to ensure safety while planning:
\begin{thm}\label{thm:time_horizon}
Let $T_\text{sense}$, $\tau_\text{plan}$, and $\tau_\text{stop}$ be as in Assumptions \ref{ass:sense},\ref{ass:plan_and_scan}, and \ref{ass:brake}, respectively.
Suppose the vehicle plans a trajectory every $\tau_\text{plan}$ and there exists a $ T \geq 0$ such that:
\begin{align}
\tau_\text{plan} + \tau_\text{stop} &\leq T  \\
T + \tau_\text{plan} &\leq T_\text{sense},
\end{align}
and a safe trajectory over the time horizon $[0,T]$. Then, there exists a safe trajectory for all $t > T$.
\end{thm}
\begin{proof}
This theorem follows directly from the assumptions, which we can use to construct the worst-case scenario.
At time $t = 0$, let the vehicle be following a trajectory that has been previously identified as safe for the time interval $[0,T]$.
Suppose trajectories are planned as described in Remark \ref{rem:plan}.
Obstacles that could cause collisions fall into one of three cases. 
It is sufficient to show that none of these cases can cause a collision.

\noindent {\bf Case 1:} Any obstacle that can be reached within time $T$.
These obstacles will already be avoided since the trajectory starting at time $t = 0$ is safe for the entire duration $[0,T]$ by definition, even if the vehicle must come to a stop.

\noindent {\bf Case 2:} Any obstacle that can be reached at a time $t \in (T,T_\text{sense}]$.
These obstacles are sensed according to Assumption \ref{ass:sense}.
The worst-case scenario, in this instance, is an obstacle that can be reachable infinitesimally after $t = T$.
Recall that the vehicle is replanning while traveling a duration of $\tau_\text{plan}$, so the plan beginning from time $t = \tau_\text{plan}$ incorporates the positions of all sensed obstacles by adjusting their relative position using the initial condition for the upcoming trajectory.
In addition, this plan incorporates the error that accumulates over the trajectory before $\tau_\text{plan}$ because the FRS uses Assumption \ref{ass:error_func}.
Therefore, at time $t = \tau_\text{plan}$, no obstacle is reachable within a time horizon less than $T - \tau_\text{plan}$, which is greater than or equal to $\tau_\text{stop}$.
Consequently, while traversing the first time interval $[0,\tau_\text{plan}]$, the vehicle determines whether or not it must begin stopping, or if a non-braking trajectory exists, at time $t = \tau_\text{plan}$.
As a result, at $t = \tau_\text{plan}$, the obstacle falls into Case 1.

\noindent {\bf Case 3:} Any obstacle that can be reached at a time $t > T_\text{sense}$.
We assume that the vehicle is detecting obstacles continually, as per Assumption \ref{ass:sense}.
Therefore, by similar logic to Case 2, the obstacle is no closer than $T_\text{sense} - \tau_\text{plan}$ at time $t = \tau_\text{plan}$.
The obstacle thus falls into Case 2 for the trajectory starting at $t = \tau_\text{plan}$.
\end{proof}


\subsubsection{FRS Computation}
\label{subsubsec:FRS_computation}
 
To compute the FRS one can solve the following linear program over the space of functions, which is adapted from \cite[Section 3.3, Program $(D)$]{majumdar2014convex}:
\begin{flalign} 
		& & \underset{v,w,q}{\text{inf}} \hspace*{0.25cm} & \int_{X_s \times K} w(x_s,k) ~ d\lambda_{X_s \times K} && (D) \nonumber \\
		& & \text{s.t.} \hspace*{0.25cm} & \Lf v(t,x_s,k) + q(t,x_s,k) \leq 0, && \text{on } [0,T] \times X_s \times K \nonumber\\
        & & & \Lg v(t,x_s,k) + q(t,x_s,k) \geq 0, && \text{on } [0,T] \times X_s \times K  \nonumber \\
        & & & -\Lg v(t,x_s,k) + q(t,x_s,k) \geq 0, && \text{on } [0,T] \times X_s \times K  \nonumber \\
        & & & q(t,x_s,k) \geq 0, && \text{on } [0,T] \times X_s \times K  \nonumber \\
        & & & -v(0,x_s,k) \geq 0, && \text{on } X_0 \times K  \nonumber \\
        & & & w(x_s,k) \geq 0,&& \text{on } X_s \times K \nonumber \\
        & & & w(x_s,k) + v(t,x_s,k) - 1 \geq 0, && \text{on } [0,T] \times X_s \times K \nonumber 
\end{flalign}
The given data in this problem are $f,g,X_s,X_0,K$ and $T$ and the infimum is taken over $(v,w,q) \in C^1\left([0,T] \times X_s \right) \times C(X_s) \times C([0,T] \times X_s)$. 

Notice first that the $1$-superlevel set of feasible $w$'s in $(D)$ are outer approximations of $\X\frs$:
\begin{lem}\label{lem:feasible_w}
Let $(v,w,q)$ be feasible functions to $(D)$, then $\X\frs \subset \big\{(x_s,k) \in X_s \times K \mid w(x_s,k) \geq 1 \big\}$.
\end{lem}
\begin{proof}
Suppose $(x_s,k) \in \X\frs$, then $\exists~\tau \in [0,T]$, $d:[0,T] \mapsto [-1,1]$ and $x:[0,T] \mapsto X_s$ such that $\dot{x}(t) = f_s(t,x(t),k) + g(t,x(t),k)d(t)$ for a.e. $t\in [0,T]$ with $x(\tau) = x_s$ and $x(0) \in X_0$.
Observe that:
\begin{align*}
v(\tau,x(\tau),k) &= v(0,x(0),k) + \int\limits_0^{\tau} \left(\Lf v(t,x(t),k) + \Lg v(t,x(t),k) d(t) \right) dt \\
&\geq v(0,x(0),k) + \int_0^{\tau} \left( \Lf v(t,x(t),k) + q(t,x(t),k) \right) dt \\
&\geq v(0,x(0),k),
\end{align*}
where the first equality follows from the Fundamental Theorem of Calculus, the second inequality follows from noticing that $\text{sup}_{d(t) \in [-1,1]} |\Lg v(t,x(t),k) d(t) | = |\Lg v(t,x(t),k) |$ and that $  |\Lg v(t,x(t),k) | \geq q(t,x(t),k) $ due to the constraints in $(D)$, and the final inequality follows from the constraints in $(D)$. 
Notice the desired result follows by noticing that $v(0,x(0),k) \geq 0$ for $x(0) \in X_0$ and applying the last constraint in $(D)$.
\end{proof}

The result of this lemma is that the $1$-superlevel set of any feasible $w$ contains all possible trajectories of the high fidelity vehicle model $f$ beginning from $X_0$.
In fact, one can prove that the solution to this infinite dimensional linear program allows one to compute $\X\frs$:
\begin{thm}{\cite[Theorem 3.5]{majumdar2014convex}}
There is a sequence of feasible solutions to $(D)$ whose $w$-component converges from above to an indicator function on $\X\frs$ in the $L^1$ norm and almost uniformly.
\end{thm}
In Section \ref{subsec:frs_imp} we describe how to generate feasible solutions to $(D)$ using semidefinite programming over polynomial representations of continuous functions.
Once this computation is completed once offline, the result can be used online by performing real-time set intersection as described next.

\subsection{Set Intersection}
\label{subsec:set_intersection}

Any outer-approximation to $\X\frs$ can be intersected with obstacles in the state space to return the trajectory parameters that could result in a collision. 
The complement of the set returned by this intersection are then parameters which can be safely followed by the high fidelity vehicle model.
To describe this process, consider a $w$ generated as a solution to $(D)$. 
One can construct a closed inner approximation to the set of safe gains, denoted $K_\text{safe} \subset K$ by solving the following optimization problem:
\begin{flalign} 
		& & \underset{h}{\text{sup}} \hspace*{0.25cm} & \int_K h(k)d\lambda_K && \label{eq:set_int_program} \\
		& & \text{s.t.} \hspace*{0.25cm} & 1 - w(x_s,k) - h(k) \geq 0, && \text{on } X_\text{obs} \times K \nonumber \\
        & & & h(k) \leq 1, && \text{on } K \nonumber
\end{flalign}
where $w$ is the given data and the supremum is taken over $h \in C(K)$. 

To appreciate how his optimization problem generates an inner approximation to the set of safe gains, consider a trajectory parameterization, $k$, that generates a trajectory which intersects an obstacle as illustrated in Figure \ref{fig:set_int_ex}.
In this instance, there exists some $x_s$ in $X_\text{obs}$ such that $w(x_s,k) \geq 1$. 
In that instance the first constraint in the previous optimization problem requires that $h(k)$ be less than zero. 
This observation is formalized in the next lemma:
\begin{lem} \label{lem:feasible_h}
Let $h$ be a feasible function to Equation \eqref{eq:set_int_program}. Then $\{k \in K \mid h(k) \geq 0 \} \subset K_\text{safe}$.
\end{lem}

By exploiting the polynomial outer approximation of the $\X\frs$ which is generated in Section \ref{subsec:frs_imp} one can formulate a convex optimization method to compute a closed, inner approximation to $K_\text{safe}$.
This optimization method is described in Section \ref{subsec:set_int_imp}. 
Note that if $T$ is chosen as in Theorem \ref{thm:time_horizon}, then as a result of Assumption \ref{ass:brake}, $K_\text{safe} \neq \emptyset$.

\begin{figure}
\centering
\includegraphics[width=1\columnwidth]{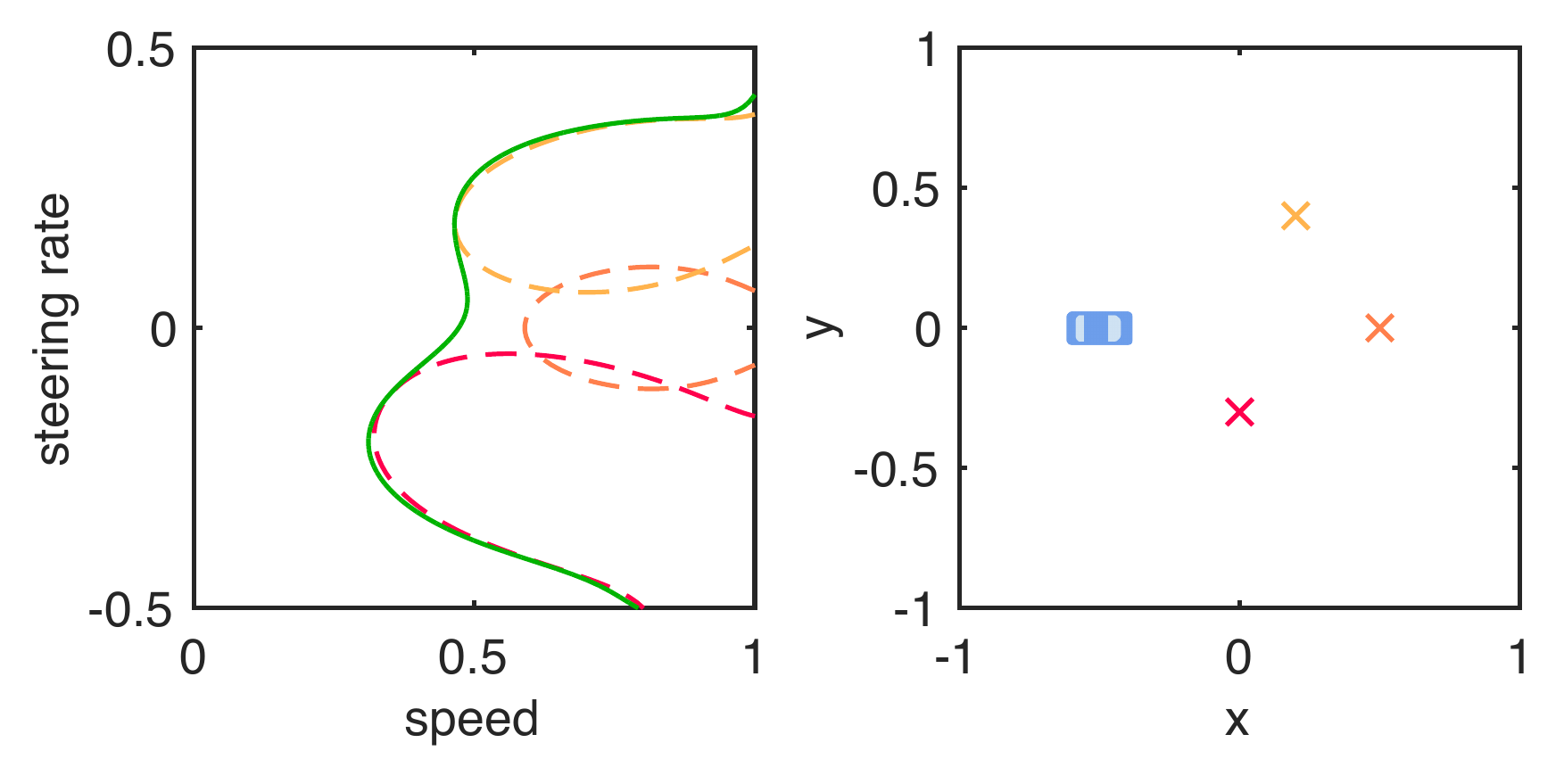}
\caption{An example of set intersection, with $X_\text{obs}$ chosen as three points in $X_s$. On the right is the $(x,y)$ subspace of $X_s$ with each point obstacle shown, and the vehicle plotted in blue. On the left is the trajectory parameter space $K$, with three dashed-line contours containing an outer approximation of the trajectory parameters that would cause a collision with each point shown on the right (the colors match between the points and contours). The set intersection step returns the subset of $K$ shown by the green contour, which outer approximates all trajectory parameters that could result in collisions with any of the three points in $X_\text{obs}$. Therefore, $K_\text{safe}$ is inner-approximated.}
\label{fig:set_int_ex}
\end{figure}

\subsection{Trajectory Optimization}
\label{subsec:traj_opt}
Once $K_\text{safe}$ is determined for an obstacle, one can optimize the original dynamics directly over the set. 
Specifically, if a user specifies a continuously differentiable cost function, $J: K \to \R$ then one can optimize this cost function with a constraint that enforces that $k \in K_{\text{safe}}$.
Since each $k$ corresponds to a $u_k$ one can optimize directly over safe trajectories of the high fidelity model that travel to some waypoint, minimize total acceleration along a trajectory, match a desired system speed, etc.

This constrained nonlinear optimal control problem can be solved in a variety of different ways via collocation, solving a variational equation, sampling, or using convex relaxations \cite{zhao2016control}. 
If this optimization program is unable to conclude within $\tau_{\text{plan}}$, then one can always apply a braking maneuver which always exists as described in Section \ref{subsec:set_intersection}.



\section{IMPLEMENTATION}
\label{sec:implementation}

This section describes the specific implementations of the FRS Computation, the set intersection, and the trajectory optimization described in Section \ref{sec:problem_formulation}.
For implementation, we require the following assumption:
\begin{assum}
The functions $f_s$ is a polynomial in $\R[t,x_s,k]$, and the error function $g$ is in $\R[t,x_s]$. In addition, the sets $K$, $X_0$, $X_s$, and $X_\text{obs}$ are semialgebraic, defined by polynomials in their respective spaces as in Equations \eqref{eq:ass_2_1} - \eqref{eq:ass_2_3} and \eqref{eq:X_obs}.
\end{assum}

\subsection{FRS Computation}
\label{subsec:frs_imp}

To solve $(D)$, we construct a sequence of convex sums-of-squares programs by relaxing the continuous function in $(D)$ to polynomial functions with degree truncated to $2l$. 
The inequality constraints in $(D)$ then transform into SOS constraints which then transforms $(D)$ into a semidefinite program\cite{parrilo2000structured}
To formulate this problem, let $h_{\tau} = t(T - t)$ and define $Q_{2l}(h_{\tau},h_{X_1},\ldots,h_{X_{n_{X_s}}},h_{K_1},\ldots,h_{K_{n_K}}) \subset \mathbb{R}_{2l}[t,x_s,k]$ to be the set of polynomials $q \in \mathbb{R}_{2l}[t,x_s,k]$ (i.e. of total degree less than $2l$) expressible as:
\begin{equation}
  q = s_0 + s_1 h_{\tau} + \sum_{i=1}^{n_{X_s}} s_{i+2} h_{X_i} + \sum_{i=1}^{n_{K}} s_{i+{n_{X_S}}+2} h_{K_i},
\end{equation}
for some polynomials $\{s_i\}_{i=0}^{n_{X_s}+n_{K}+1} \subset \R_{2k}[t,x,k]$ that are sums of squares of other polynomials, where we have dropped the dependence on $t,x,$ and $k$ in each polynomial for the sake of convenience.
Note that every such polynomial is non-negative on $[0,T] \times X_s \times K$ \cite[Theorem 2.14]{lasserre2009moments}.
Define $Q_{2l} (h_{0_1},\ldots,h_{0_{n_0}},h_{K_1},\ldots,h_{K_{n_K}}) \subset \mathbb{R}_{2l}[x_s,k]$, $Q_{2l} (h_{X_1},\ldots,h_{X_{n_{X_s}}},h_{K_1},\ldots,h_{K_{n_K}}) \subset \mathbb{R}_{2l}[x_s,k]$, $Q_{2l} (h_{\text{obs}_1},\ldots,h_{\text{obs}_{n_\text{obs}}},h_{K_1},\ldots, h_{K_{n_K}}) \subset \mathbb{R}_{2l}[x_s,k]$, and $Q_{2l}(h_{K_1},\ldots, h_{K_{n_K}}) \subset \mathbb{R}_{2l}[k]$ similarly.

Employing this notation, the $l$-th relaxed semidefinite programming representation of $(D)$, denoted $(D_l)$, is defined as:
\begin{flalign} 
		& & \underset{v,w,q}{\text{inf}} \hspace*{0.25cm} & y_{X_s \times K}^T \textrm{vec}(w) && (D_l) \nonumber \\
		& & \text{s.t.} \hspace*{0.25cm} & -\Lf v - q \in Q_{2l}(h_{\tau},h_{X_1},\ldots,h_{X_{n_{X_s}}},h_{K_1},\ldots,h_{K_{n_K}}) && \nonumber\\
        & & & \Lg v+ q \in Q_{2l}(h_{\tau},h_{X_1},\ldots,h_{X_{n_{X_s}}},h_{K_1},\ldots,h_{K_{n_K}}) && \nonumber \\
        & & & -\Lg v + q\in Q_{2l}(h_{\tau},h_{X_1},\ldots,h_{X_{n_{X_s}}},h_{K_1},\ldots,h_{K_{n_K}}) &&  \nonumber \\
        & & & q \in Q_{2l}(h_{\tau},h_{X_1},\ldots,h_{X_{n_{X_s}}},h_{K_1},\ldots,h_{K_{n_K}}) &&  \nonumber \\
        & & & -v(0,\cdot) \in Q_{2l} (h_{0_1},\ldots,h_{0_{n_0}},h_{K_1},\ldots,h_{K_{n_K}}) &&  \nonumber \\
        & & & w \in Q_{2l} (h_{X_1},\ldots,h_{X_{n_{X_s}}},h_{K_1},\ldots,h_{K_{n_K}})&&  \nonumber \\
        & & & w + v - 1 \in Q_{2l}(h_{\tau},h_{X_1},\ldots,h_{X_{n_{X_s}}},h_{K_1},\ldots,h_{K_{n_K}}) && \nonumber 
\end{flalign}
where the infimum is taken over the vector of polynomials $(v,w,q) \in \R_l[t,x_s,k] \times \R_l[x_s,k] \times \R_l[t,x_s,k]$. 

If $w_l$ denotes the $w$-component of the solution to $(D_l)$, one can prove that $w_l$ converges from above to an indicator function on $\X\frs$ in the $L^1$ norm \cite[Theorem 6]{majumdar2014convex}.
Importantly since each such $w_l$ is feasible with respect to the constraints in $(D)$, one can apply the result of Lemma \ref{lem:feasible_w} to prove that the $1$-superlevel set of $w_l$ is an outer approximation to $\X\frs$.  

\subsection{Set Intersection}\label{subsec:set_int_imp}

In this section, we present a semidefinite program to generate a closed inner-approximation to $K_\text{safe}$ and comment on computing $\tau_\text{plan}$.
Using the notation developed in Section \ref{subsec:frs_imp}, consider the following semidefinite formulation of Equation \eqref{eq:set_int_program}:
\begin{flalign} 
		& & \underset{h}{\text{sup}} \hspace*{0.25cm} & y_{K}^T \textrm{vec}(h) && \label{eq:SDP_h}  \\
		& & \text{s.t.} \hspace*{0.25cm} & 1 - w - h \in Q_{2l} (h_{\text{obs}_1},\ldots,h_{\text{obs}_{n_\text{obs}}},h_{K_1},\ldots, h_{K_{n_K}}) && \nonumber\\
        & & & 1 - h \in Q_{2l}(h_{K_1},\ldots, h_{K_{n_K}}) && \nonumber 
\end{flalign}
where the given data is any feasible $w \in \R_l[x_s,k]$ to $(D_l)$ and the supremum is taken over polynomials $h \in \R_l[k]$. 
Note again that since the constraints in this optimization problem are sufficient to ensure positivity on the required sets in Equation \eqref{eq:set_int_program}, one can apply the result of Lemma \ref{lem:feasible_h} to prove that the $0$-superlevel set of $h_l$ is an inner approximation to $K_{\text{safe}}$.





Solving such an SDP, while possible in real time, is still computationally expensive.
Compared to the trajectory optimization step, which takes milliseconds, the set intersection step can take several tenths of a second per obstacle, and is thus the primary contributor to $\tau_\text{plan}$.
The planning time increases with the number of dimensions of $X_\text{obs}$ and with the number of obstacles.
We find $\tau_\text{plan}$ by precomputing an FRS over an arbitrary, short time horizon $T$, then empirically determining $\tau_\text{plan}$ by performing set intersection with representative obstacles and environments.

\subsection{Trajectory Optimization}\label{subsec:traj_opt_imp}
We implement trajectory optimization by directly passing $h_l$ from solving \eqref{eq:SDP_h} as a nonlinear constraint to an off-the-shelf gradient-descent solver, such as MATLAB's \texttt{fmincon}.

\section{SIMULATION RESULTS}
\label{sec:examples}
This section presents simulation results of the proposed method implemented on the unicycle model example.
Each semidefinite program was prepared using a custom software toolbox and the modeling tool YALMIP \cite{lofberg2004yalmip}.
The programs are run with commercial solver MOSEK on a machine with $1$ TB availabe memory. 

\subsection{FRS Computation}
We computed the FRS for a 3$^\text{rd}$ order Taylor-expanded Dubins car as the low-fidelity model $f_s$.
Trajectories produced by this model were tracked by the unicycle model from Equation \eqref{eq:big_dyn} as the high-fidelity model $f$.
The vehicle's representation as an initial distribution $X_0 \subset X_s$, was a rectangle of length $0.2$ [m] in $x$ and width $0.1$ [m] in $y$, at $0^\circ$ initial heading, and centered at $x=-0.75$ and $y=0$.
This is the same vehicle representation shown in all previous figures.

We chose $\tau_\text{stop} = \tau_\text{plan} = 0.5$ [s], so $T = 1$ [s].
The stopping time can be seen in Figure \ref{fig:error_dynamics}. 
The FRS computation took 79 hours and used a maximum of 150 GB of memory 

\subsection{Set Intersection and Trajectory Planning}

We used the precomputed FRS for safe trajectory planning in $1000$ simulated trials in MATLAB on the aforementioned machine.
For each trial, the vehicle began at the same initial location and heading, surrounded by $1-10$ randomized obstacles and a randomly-located goal to reach.
The vehicle's initial speed, and the desired speed to maintain for the duration of the trial, were randomly chosen between $0.25$ and $0.75$ [m/s].

Obstacles were represented as line segments between $0.1$ and $0.2$[m] in length, with random location and orientation.
The obstacles were always placed between the vehicle and the goal.
We checked for crashes conservatively for each trial, by inspecting if any obstacle was within a circle circumscribing the rectangular vehicle at any point of the vehicle's trajectory. 
Using this method, \emph{no crashes were detected in any trial}.
Out of all the trials, $82\%$ reached the goal, and $15\%$ performed an emergency braking maneuver (by setting $v_\text{des} = 0$). 
The remaining 3\% hit a simulation iteration limit.
Examples of the vehicle's path from a randomly-generated trial and from two constructed emergency braking cases are shown in Figure \ref{fig:example_trial}.

\begin{figure}
\centering
\includegraphics[width=1\columnwidth]{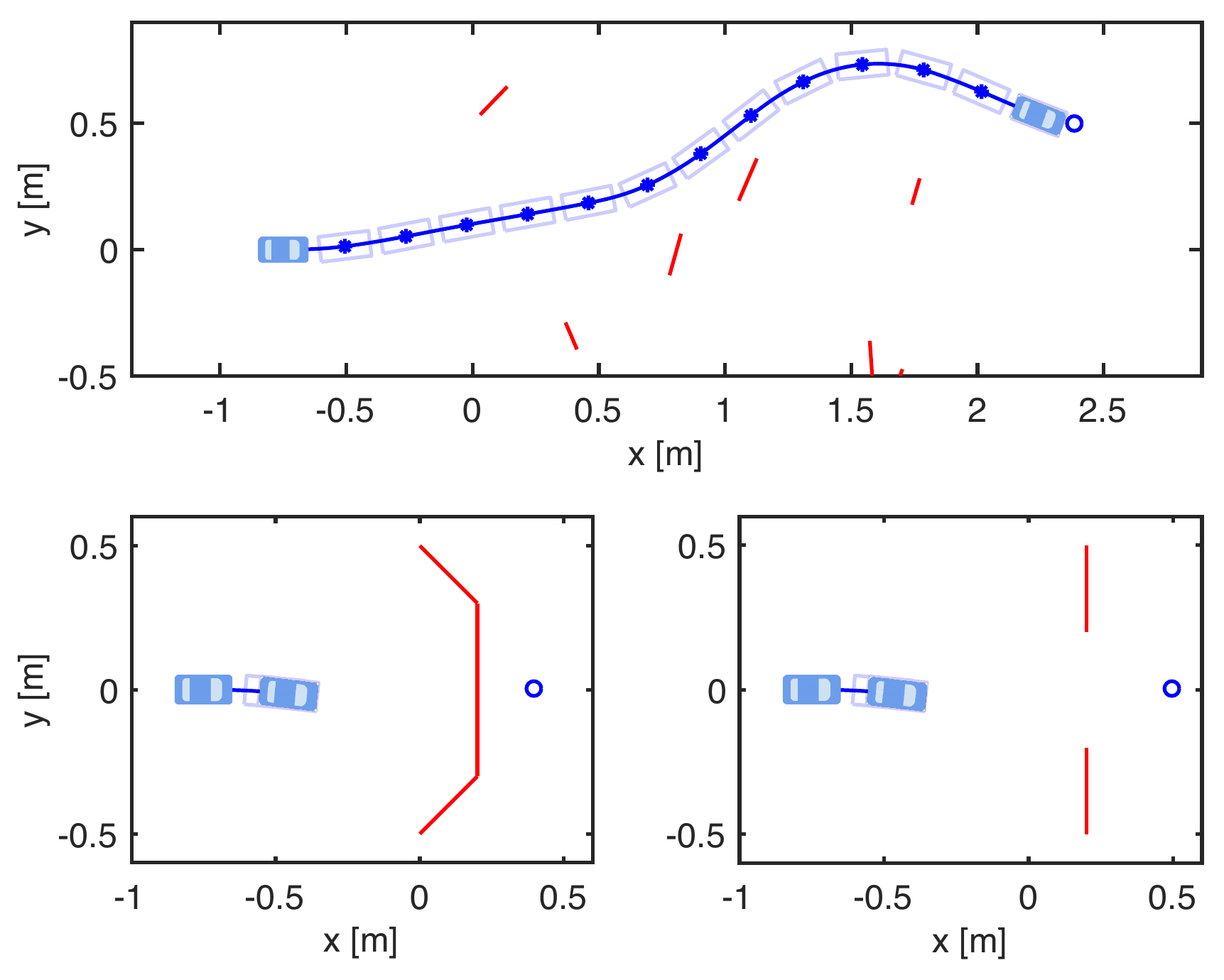}
\caption{The top subplot shows an example result out of the $1000$ trials.
This trial used eight randomly-generated obstacles.
The vehicle begins on the left at $x = -0.75$ and reaches a randomly-generated goal near $(2.5, 0.5)$, plotted as a blue circle.
Every $\tau_\text{plan} = 0.5$[s], the vehicle replans its trajectory, shown by an asterisk plotted on the global trajectory in blue.
The bounding box of the vehicle at each planning step is shown as a grey rectangle. In the bottom-left subplot, an obstacle was constructed between the vehicle and the goal, forcing an emergency braking maneuver. In the bottom-right subplot, an obstacle was constructed with a hole that would allow the vehicle to pass, but the set intersection result is overly conservative, resulting in a braking maneuver.}
\label{fig:example_trial}
\end{figure}

Currently, our implementation cannot consistently achieve $\tau_\text{plan} = 0.5$ [s].
Consequently, instead of replanning and driving simultaneously, we pause time every 0.5 [s] of the simulation to guarantee that the vehicle can finish replanning.
In a physical implementation, if $\tau_\text{plan}$ is exceeded, then the vehicle must emergency brake; recall that a safe braking trajectory is always available.
As shown in Figure \ref{fig:planning_time_vs_Nobs}, $\tau_\text{plan}$ scales linearly with the number of obstacles.

\begin{figure}
\centering
\includegraphics[scale=0.45,trim={1cm 6cm 1cm 7cm},clip]{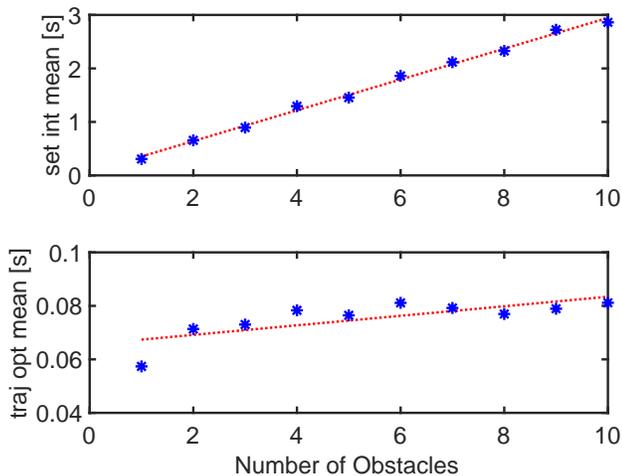}
\caption{The mean set intersection time (top) and trajectory optimization time (bottom) versus the number of obstacles. Over the $1000$ trials, each number of obstacles from $1$ to $10$ was used for $100$ trials. Notice that set intersection takes up to $3$[s], and scales with the number of obstacles. On the other hand, the trajectory optimization takes around $80$ [ms] and has low correlation with number of obstacles.}
\label{fig:planning_time_vs_Nobs}
\end{figure}




\section{CONCLUSION}
\label{sec:conclusion}


This paper presents a method to plan safe trajectories with obstacle avoidance for autonomous vehicles.
This approach is able to guarantee safety in arbitrary environments for multiple, static obstacles.
The method begins with computing the forward reachable set (FRS) of parameterized trajectories that a vehicle can realize.
This set is computed in continuous space and time, and is robust to model uncertainty between the dynamics of the vehicle's mid- and low-level controllers.

As an example, we use a kinematic Dubin's car and dynamic unicycle model as low- and high-fidelity models.
At runtime, the FRS is intersected with obstacles to eliminate unsafe trajectories, and an optimal trajectory is chosen from the remaining, safe trajectories.
This method is proven to ensure vehicle safety for present and future static obstacles by considering the time required for path planning, the time required to stop the vehicle, and the error between the low- and high-fidelity models.
One thousand simulations with randomly-located obstacles were run to show the effectiveness of this method. 


The next step in applying this method is to expand the error function $g$ beyond modeling uncertainty.
To reduce the conservativeness of the approach, we plan to explore extension to the FRS computation that incorporate confidence level sets \cite{mohan2016convex,holmes2016convex}.
The error function $g$ can also be improved by considering time variation of trajectory parameters across planning steps by posing the FRS computation as a hybrid problem \cite{shia2014convex}.
Finally, we plan to use convex optimization to find a global solution to the nonlinear trajectory optimization problem at each planning step \cite{zhao2016control,zhao2017optimal}.

\bibliography{conditionsforsafetybib}
\bibliographystyle{plain}

\end{document}